\documentclass[11pt,letterpaper]{article}
\usepackage{graphicx}

\usepackage{amsmath}

\newtheorem{theorem}{Theorem}[section]

\newtheorem{proposition}[theorem]{Proposition}
\newtheorem{lemma}[theorem]{Lemma}
\newtheorem{claim}[theorem]{Claim}

\newenvironment{proof}[1][Proof]{\textbf{#1.} }{\ \rule{0.5em}{0.5em}}

\begin{document}

\title{On bin packing with clustering and bin packing with delays}

\author{Leah
Epstein\thanks{ Department of Mathematics, University of Haifa,
Haifa, Israel. \texttt{lea@math.haifa.ac.il}. Partially supported
by a grant from GIF - the German-Israeli Foundation for Scientific
Research and Development (grant number I-1366-407.6/2016).}}

\date{}

\maketitle

\begin{abstract}
We continue the study of two recently introduced bin packing type
problems, called bin packing with clustering, and online bin
packing with delays. A bin packing input consists of items of sizes not
larger than $1$, and the goal is to partition or pack them into
bins, where the total size of items of every valid bin cannot
exceed $1$.

In bin packing with clustering, items also have colors associated
with them. A globally optimal solution can combine items of
different colors in bins, while a clustered solution can only pack
monochromatic bins. The goal is to compare a globally optimal
solution to an optimal clustered solution, under certain
constraints on the coloring provided with the input.  We show
close bounds on the worst-case ratio between these two costs,
called {\it the price of clustering}, improving and simplifying
previous results. Specifically, we show that the price of
clustering  does not exceed $1.93667$, improving over the previous
upper bound of $1.951$, and that it is at least $1.93558$,
improving over the previous lower bound of $1.93344$.

In online bin packing with delays, items are presented over time.
Items may wait to be packed, and an algorithm can create a new bin
at any time, packing a subset of already existing unpacked items
into it, under the condition that the bin is valid. A created bin
cannot be used again in the future, and all items have to be
packed into bins eventually. The objective is to minimize the
number of used bins plus the sum of waiting costs of all items,
called delays. We build on previous work and modify a simple
phase-based algorithm. We combine the modification with a careful
analysis to improve the previously known competitive ratio from
$3.951$ to below $3.1551$.
\end{abstract}

\section{Introduction}
In bin packing problems, a set of items $I$ is given, where each
item has a rational size in $[0,1]$\footnote{We allow zero sizes
as they are meaningful in the second problem which we study.}. The
goal is to partition these items into subsets called bins, where
the total size for each bin does not exceed $1$. We use the term
load of a bin for the sum of sizes of its items.  The process of
assigning an item to a bin is called packing, and in such a case
we say that the item is packed into the bin.

We study two bin packing problems. The first problem is called
{\it bin packing with clustering}. In this problem, every item has
a second attribute, called a cluster index or a color. A global
solution is one where items are packed without considering their
clusters, i.e., it is a solution of the classic bin packing
problem for this input. A clustered solution is one where every
cluster or color must have its own set of bins, and items of
different clusters cannot be packed into a common bin. To avoid
degenerate cases, an assumption on the input is enforced.
Specifically, it is assumed that every cluster is sufficiently
large, and an optimal solution for each cluster has at least three
bins. The problem was introduced by Azar et al. \cite{AESV}. It
was shown \cite{AESV} that replacing this assumption with the
weaker one where clusters have at least two bins makes the problem
less meaningful. The goal is to compare optimal solutions, that
is, to compare an optimal clustered solution to an optimal global
solution, also called a globally optimal solution. We are
interested in the worst-case ratio over all valid inputs, and this
ratio is called {\it price of clustering}. From an algorithmic
point of view, the goal is to design an approximation algorithm
for which it is not allowed to mix items of different clusters,
while the algorithm still has a good approximation ratio compared
to a globally optimal solution. For applications of this problem
in the field of massive data sets, see \cite{AESV}.

The results of \cite{AESV} show that the price of clustering
(under the assumption above) is strictly below $2$, and more
specifically, it is at most $1.951$. The methods used to prove
this are based on an auxiliary graph comparing the two different
optimal solutions, and a linear program capturing the properties
of worst-case inputs. A computer assisted proof was used to find
an upper bound on the price of clustering. A lower bound of
$1.93344$ was provided as well in the same work. This problem is
closely related to {\it batched bin packing}
\cite{GJYbatch,Do15,BBDGT,Epstein16}. This is a semi-online
problem where items are presented in a number of batches, where
every batch is to be packed before the next batch is presented.
There are two variants, depending on whether bins opened for
earlier batches can be used for the current batch. The variant
where every batch has its own bins, and the packing is compared to
an optimal (offline) one where items of different batches can
still be combined into bins together is closely related to our
work. It is mentioned in \cite{AESV} that if every cluster is
arbitrarily large such that its optimal cost grows to infinity,
then the price of clustering decreases to approximately $1.691$
(we discuss this value
\cite{BC81,LeeLee85,GW93,EL06a,Woegin93,Epstein16} in the body of
the paper in a different context). In fact, this result regarding
the price of clustering with very large clusters follows directly
from an  earlier result for batched bin packing \cite{Epstein16}.

The second problem is {\it bin packing with delays}. In this
online problem, items are presented over time to be packed into
bins. An algorithm can decide to create a bin at any time by
selecting a subset of already existing unpacked items. The
selected subset should have total size at most $1$, and once its
bin is created, it cannot be used again for future items.
Additionally, every item $i$ has a positive monotonically
non-decreasing delay function $d_i$, and letting $t_i \geq 0$ be
the elapsed time from the arrival date of $i$ until it is packed,
the delay cost (or delay) of $i$ is $d_i(t_i)$. The objective is
to minimize the number of bins plus the total delay cost of all
input items, and the goal is to minimize this objective. For
example, if every item is assigned to a bin right when it arrives,
the delays are the smallest possible, but the number of bins may
be very large. On the other hand, if the algorithm waits until
many items arrive and it can pack them offline, the delay costs
may be very large. The problem is analyzed via the competitive
ratio, which is the worst-case ratio between the cost of an online
algorithm and an optimal offline solution (which still deals with
the input as a sequence arriving over time, but it knows the
entire sequence). Competitive algorithms should find a trade-off
between waiting for additional items to arrive and the resulting
delay costs of already existing items, and one expects to see
algorithms designed based on ski-rental type methods
\cite{Karp92,KMRS88,karlin2003dynamic}. Such methods involve
waiting until a certain cost is incurred before performing an
action that stops the accumulation of that cost. Obviously,
additional problem-specific methods are required in the design of
algorithms for problems with delay costs.

Various online combinatorial optimization problems with delays
were studied recently
\cite{emek2016online,azar2017online,bienkowski2017match},
continuing earlier studies of ski-rental type problems. Moreover,
a completely different model of bin packing with delays was
studied as well \cite{ahlroth2013online}. Offline and online bin
packing are often studied with respect to asymptotic measures
\cite{FvL,KK82,BBDEL_ESA18,BBDEL_newlb}, while here we study them
via absolute measures, as in previous work on the specific
problems we study, where the absolute measure is more appropriate
(see \cite{SL94,ftsoda,DS12,DS14} for studies of bin packing with
respect to absolute measures). The two problems studied here may
seem unrelated; one is an offline problem and the other one is a
completely different online problem. The flavor of the first
problem is not algorithmic, and the algorithmic contribution is
used in the analysis. The second problem is an online problem
where items arrive over time, and even if one designs an offline
algorithm for it, still the time axis has a major role. Since the
two problems were introduced and studied in the same work
\cite{AESV} where properties of the first one were used in the
analysis of the second one, we study them together as well. Note
that we also use properties of offline bin packing for the
analysis of the online problem, as we will pack subsets of items
at the same time, into one bin or several bins. Bin packing with
delays is a special case of the TCP acknowledgement problem
\cite{dooly2001line,karlin2003dynamic}. In this problem requests
arrive over time, and should be acknowledged at times selected by
the algorithm, where at every such time, all pending requests can
be acknowledged. The objective is the number of acknowledgement
events plus the total waiting time of all requests. Instances of
this problem are instances of bin packing with delays with zero
size items and delay costs based on the identity function (there
is also work on more general delay functions, see for example
\cite{AB05}). Using the lower bound of $2$ on the competitive
ratio of any algorithm for TCP acknowledgement, a lower bound of
$2$ is known also for the competitive ratio of any algorithm for
bin packing with delays.

In this work, we improve the bounds on the price of clustering,
and show close bounds of  $1.93667$ and $1.93558$. The upper bound
is shown via weighting functions, while the lower bound uses a
careful refinement of the previous lower bound approach, where not
only clusters with items of sizes close to $\frac 12$ are defined
with respect to the worst-case structure but also more complicated
clusters are built. We also show how the previous upper bound
result can be obtained using a simple analytical proof, and we
briefly discuss other versions (with larger clusters). We also
generalize the previous algorithm for bin packing with delays such
that its parameter can be arbitrary. Here, we apply a simple
weight based analysis to obtain a better upper bound of $3.1551$,
while the previous bound was $3.951$ \cite{AESV}. Our algorithm
does not require computation of optimal solutions, and whenever it
packs a subset of items, this is done using a greedy algorithm,
and therefore it runs in polynomial time if the delay function can
be computed easily.

\section{Price of clustering}
In this section we study the price of clustering. Note that we
consider the case where optimal costs for clusters are at least
$3$. Considering a parameter $k \geq 1$, such that the optimal
cost for every cluster is at least $k$, the cases $k=1,2$ were
fully analyzed and declared as uninteresting \cite{AESV}. For
$k=1$, the price of clustering is unbounded, as an input of very
small items may be partitioned into clusters containing single
items. For $k=2$, the price of clustering is $2$, since every
cluster may have one item slightly larger than $\frac 12$ and one
item slightly smaller than $\frac 12$, where these items cannot be
packed into one bin, while in a globally optimal solution they can
be packed in the suitable pairs. On the other hand, every bin is
full by more than half on average.  We study the most general case
where the bound on the price of clustering is strictly below $2$.
For a different parameter $k \geq 4$ one can use similar proofs to
find close bounds (the tight bounds are expected to tend to
approximately $1.691$ as $k$ grows). The lower bounds will have a
similar structure while the upper bounds will require some
modifications of the weight functions.

\subsection{A lower bound}
Our new lower bound has some similarity to the one of \cite{AESV}.
The idea was that there can be clusters with two items of sizes
just above $\frac 12$ and one item of size just below $\frac 12$,
such that no two items can be combined in one bin. In our
construction, there will also be clusters with five items of sizes
approximately $\frac 13$, such that no three items can be packed
into a bin, so at most two of them have sizes of $\frac 13$ or
less. Similarly, there will be clusters with $11$ items of sizes
approximately $\frac 16$, where no six items can be packed into a
bin. It is possible to continue and have some clusters with $13$
items of sizes approximately $\frac 17$ and clusters with $83$
items of sizes approximately $\frac 1{42}$ and so forth, but this
will increase only the sixth digit after the decimal point. As
already now the items sizes have to be defined carefully, and
calculations need to be done precisely to ensure the costs of
clusters, we do not give the details of such a construction. The
current construction can be also continued with additional very
small items, but that would also not increase the value of the
lower bound significantly.

Let $M>2$ be a large integer. Let $N>10$ be an integer parameter
divisible by $5000!\cdot 9^M$,  We construct an input where there
is a globally optimal solution with $N$ bins. The input consists
of the following items. Let $\mu>0$ be a very small value.

\begin{itemize}
\item For $1 \leq i \leq \frac N2$, a positive type $(2,i)$ item
has size $\frac 12+i\cdot \mu$.

There is one such item for
every $i=1,2,\ldots,\frac N2$.

\item For $1 \leq i \leq \frac N2$, a negative type $(2,i)$ item
has size $\frac 12-i\cdot \mu$.

There are one such item for
every $i=1,2,\ldots,\frac N2$.

\item There are $\frac N2$ type $2$ items, each of size $\frac
12+\frac N2 \cdot \mu$.

\item For $1 \leq i \leq M$, a positive type $(3,i)$ item has size
$\frac 13+3^{N+3i}\cdot \mu$. The number of positive type $(3,i)$
items is $\frac{2N}{15} \cdot (\frac 59)^{M-i}$.

\item For $1 \leq i \leq M$, a negative type $(3,i)$ item has size
$\frac 13-3^{N+3i-1}\cdot \mu$.  The number of negative type
$(3,i)$ items is $\frac{4N}{45} \cdot (\frac 59)^{M-i}$.

\item For $2 \leq i \leq M$, a positive type $(6,i)$ item has size
$\frac 16+3^{N+3i-1}\cdot \mu-N\cdot \mu$. The number of positive
type $(6,i)$ items is: \\ $\frac{4N}{45} \cdot (\frac
59)^{M-i}=\frac{4N}{25} \cdot (\frac 59)^{M-(i-1)}$.

\item For  $1 \leq i \leq M-1$, a negative type $(6,i)$ item has
size $\frac 16-3^{N+3i}\cdot \mu-N\cdot \mu$. The number of
negative type $(6,i)$ items is $\frac{2N}{15} \cdot (\frac
59)^{M-i}$.

\item A type $7$ item has size $\frac 17+\mu$. The number of such
items is $\frac{2N}{15}$.

\item A type $43$ item has size $\frac 1{43}+\mu$. The number of
such items is $\frac{2N}{15}$.

\item A type $1807$ item has size $\frac 1{1807}+\mu$. The number
of such items is $\frac{2N}{15}$.

\end{itemize}

It is obvious that there is no global solution whose cost is below
$N$. A globally optimal solution is defined as follows. For
$i=1,2,\ldots,\frac N2$, there is a bin with one positive type
$(2,i)$ item and one negative type $(2,i)$ item, where the total
size for such a pair of items is $1$.

Every bin out of the remaining $\frac N2$ bins has a type $2$
item, so the remaining space of such a bin is $\frac 12-\frac
N2\cdot \mu$, and this is where all other items are packed.
Specifically, every positive or negative type $(3,i)$ is packed
into a different such bin. The number of such items is
$\frac{10N}{45} \cdot \sum_{i=1}^M \left(\frac 59\right)^{M-i}  $,
where
$$\sum_{i=1}^M \left(\frac 59\right)^{M-i}=\sum_{i=0}^{M-1} \left(\frac
59\right)^{i}<\sum_{i=0}^{\infty} \left(\frac 59\right)^{i}=\frac
94 \ .
$$ Thus the number of these items is below $\frac{10N}{45}\cdot
\frac 94=\frac N2$.

The $\frac {2N}{15}$ bins with the largest positive type $(3,i)$
items, that is, each of the bins with a positive type $(3,M)$
item, contains also one item of each type out of $7$, $43$, and
$1807$. The total size for such a bin is $$\frac 12+\frac N2 \cdot
\mu + \frac 13+3^{N+3M}\cdot \mu +\frac 17+\frac 1{43}+\frac
1{1807}+3\mu<1 \ , $$ for a sufficiently small value of $\mu$.

For $i<M$, every positive type $(3,i)$ items is packed with a
negative type $(6,i)$ item. For $i>1$, every negative type $(3,i)$
items is packed with a positive type $(6,i)$ item. The total size
of items of every such bin is exactly  $1-\frac N2\cdot \mu$. As
for negative type $(3,1)$ items, they are not combined with
additional items and the loads of their bins are approximately
$\frac 56$. Thus, all items are packed into $N$ bins as claimed.

Next, we split items into clusters, and we find the optimal cost
for every cluster (in particular we will see that it is at least
$3$ as it is required for a valid input). We will calculate the
total number of bins for the optimal clustered solution. In this
input, every cluster will have items of similar sizes.

\medskip

\noindent  {\bf 1.}  Type $1807$ items are split into subsets of
$3613$ items each. Since a bin can contain at most $1806$ such
items while $\mu$ is sufficiently small such that $1806$ items can
be packed into a bin, an optimal solution has three bins. Thus, as
there are $\frac{2N/15}{3613}$ clusters, the contribution to the
cost is $\frac{2N}{18065}$.

\smallskip

\noindent {\bf 2.} The calculation for type $43$ items is similar
to the previous one. Here a cluster will have $85$ items, there
are $\frac{2N/15}{85}$ clusters, the contribution to the cost is
$\frac{2N}{425}$.

\smallskip

\noindent {\bf 3.} The calculation for type $7$ items is similar
to the last two calculations. Here a cluster will have $13$ items,
there are $\frac{2N/15}{13}$ clusters, the contribution to the
cost is $\frac{2N}{65}$.

\smallskip

\noindent {\bf 4.} There are $\frac N2-1$ clusters, each
containing one type $2$ item, and for some value of $i$ ($1 \leq i
\leq \frac N2-1$), a positive type $(2,i+1)$ item and a negative
type $(2,i)$ item. As no two items of one cluster have total size
of $1$ or less, the optimal cost for each cluster is $3$. The
contribution to the cost is therefore $3(\frac N2-1)$. The
remaining three items, a type $2$ item, a positive type $(2,1)$
item, and a negative type $(2,\frac N2)$ item are added to one of
the clusters, which does not decrease its optimal cost.
\smallskip

\noindent {\bf 5.} For every $1 \leq i \leq M$, there is a cluster
consisting of five items as follows: three positive type $(3,i)$
items and two negative type $(3,i)$ items. The number of clusters
for a fixed value of $i$ is $\frac{2N}{45} \cdot \left(\frac
59\right)^{M-i}$. Since $$\frac 13+3^{N+3i}\cdot \mu+2\left(\frac
13-3^{N+3i-1}\cdot \mu\right)=1+3^{N+3i-1}\cdot \mu>1 \ , $$ no
three items fit into one bin, and an optimal solution for every
cluster uses three bins. The contribution to the cost is
$$3 \cdot \frac{2N}{45} \cdot \sum_{i=1}^M\left(\frac
59\right)^{M-i}=\frac{2N}{15}\cdot \sum_{i=0}^{M-1}\left(\frac
59\right)^{i}=\frac{2N}{15}\cdot\frac{1-\left(\frac
59\right)^M}{4/9}=0.3N\left(1-(\frac 59)^M\right) \ . $$
\smallskip

\noindent
{\bf 6.}  For every $1 \leq i \leq M-1$, there is a cluster consisting
of eleven items as follows: six positive type $(6,i+1)$ items and
five negative type $(6,i)$ items. The number of clusters for a
fixed value of $i$ is $\frac{2N}{75} \cdot (\frac 59)^{M-i}$.

Since $$\frac 16+3^{N+3i+2}\cdot \mu-N\cdot \mu+5(\frac
16-3^{N+3i}\cdot \mu-N\cdot \mu)=1+4 \cdot 3^{N+3i}\cdot
\mu-6N\mu>1 \,
$$ no six items fit into one bin, and an optimal solution for
every cluster uses three bins. The contribution to the cost is
$$3 \cdot \frac{2N}{75} \cdot \sum_{i=1}^{M-1}(\frac 59)^{M-i}= \frac{2N}{25} \cdot \frac 59 \cdot \sum_{i=0}^{M-2}(\frac 59)^{i}=\frac{2N}{45}\cdot \sum_{i=0}^{M-2}(\frac 59)^{i}=\frac{2N}{45}\cdot\frac{1-(\frac 59)^{M-1}}{4/9}
=0.1N(1-(\frac 59)^{M-1}) \ . $$
\smallskip

\noindent
All items were assigned to suitable clusters, and the total cost of the clustered optimal solution is at
least
$$N\cdot\left(\frac{2}{18065}+\frac 2{425}+\frac{2}{65}+\frac{3}{2}-3/N+0.3\cdot(1-(5/9)^M)+0.1(1-(5/9)^{M-1}\right)  \ . $$

Letting $N$ and $M$ grow without bound, the ratio between the
costs of the two optimal solutions is approximately
1.9355858244424.

\begin{theorem}
The price of clustering is at least $1.93558$.
\end{theorem}

\subsection{Upper bounds for $\boldsymbol{k=3}$}

We will start the analysis of upper bounds with a simple analysis
of the price of clustering, yielding the bound of \cite{AESV} in a
simple way (in fact, since we use an analytic proof, we show a
value of $1.95$ rather than $1.951$). Unlike the previous proof,
we do not use auxiliary graphs or computer assisted analysis. Our
improved result of $\frac{349}{180}\approx 1.38889$ will be based
on an extension of the approach of the simple bound.

The analysis yielding the bound $1.95$ resembles the one of
Simchi-Levi for First-Fit Decreasing (FFD) \cite{SL94}. In this
algorithm, items are sorted by non-increasing size and First-Fit
(FF) is applied to this list. FF is a greedy algorithm that packed
every item into the bin of smallest index where it can be packed,
given the previously packed items, which are not smaller in the
case of FFD.

For a fixed input, let $\ell$ be the number of clusters. Let
$OPT_i$ be the number of bins in an optimal solution for the $i$th
cluster, whose input is $I_i$. We let $I$ be the set of items
$I=\bigcup_{1 \leq i \leq \ell} I_i$, where $n=|I|$. Let $OPT$ be
a globally optimal solution for $I$, as well as its cost. Let
$A_i$ be the number of bins in the output of FFD for cluster $i$.

We will use weights for the analysis. Weights allow us to bind two
solutions and compare them, using the property that the total
weight of all input items can be defined consistently.

We start with defining a simple weight function. Let
$w(x):[0,1)\rightarrow (0,1.95]$ as follows.
\begin{equation*}
w(x)=1.8 \cdot x+\begin{cases} 0.15 { \ \ \ \mbox {for} \ \ \ } x>\frac 12, \\
0 { \ \ \  \ \ \ \ \mbox {for} \ \ \ } 0 \leq x \leq \frac 12.
\end{cases}
\end{equation*}

Let $W=\sum_{j \in I}w(s_j)$, where $s_j>0$ is the size of item
$j$. Let $W_i=\sum_{j \in I_i} w(s_j)$, where $W=\sum_{i=1}^{\ell}
W_i$. Similarly, $S=\sum_{j=1}^n s_j$, and $S_i=\sum_{t \in I_i}
s_t$. An item of size strictly above $0.5$ is called large. The
next claim provides an upper bound on the total weight, based on
the value $OPT$.

\begin{claim} We have $W \leq 1.95 \cdot OPT$.\end{claim}
\begin{proof} Consider a bin $\cal{B}$ of OPT. The total size of items
is at most $1$, and there is at most one large item, which gives
at most $1.8+0.15$, since the total weight is at most  \\
$(\sum\limits_{j \in \cal{B}} 1.8\cdot s_j)+0.15
=(1.8\sum\limits_{j \in \cal{B}} s_j)+0.15 \leq 1.95 $.\end{proof}

\smallskip

The next claim holds by definition, and by the assumption on
clusters.
\begin{claim} $\sum_{i=1}^{\ell} OPT_i \leq \sum_{i=1}^{\ell} A_i$ and
$3 \leq OPT_i \leq A_i$.
\end{claim}

Consider the output bins of FFD for some input. Recall that
indexes of bins are given according to the order in which FFD
opens (first uses) them. Let all bins for an output of FFD be
called {\it inner} except for the last bin. When we say that a bin
is earlier than another bin, we mean that it has a smaller index,
and a later bin has a larger index.

\begin{claim}\label{was3} The set of bins of FFD with large items is a prefix
of its bins. Every pair of bins of FFD has total size above $1$.
If the first item of some bin has size above $\frac 1{s}$ for an
integer $s \geq 2$, then every earlier bin without any item of
size above $\frac{1}{s-1}$ has $s-1$ items of sizes in $(\frac
1s,\frac{1}{s-1}]$ (and possibly other items packed later).
\end{claim}

The first part holds because the large items are packed first (and
a pair of such items cannot share a bin). The second part and
third part hold due to the rule of opening a new bin.

\begin{claim}\label{con1} Assume that all inner bins of FFD except for
possibly one bin (called bad) have loads of at least $\frac 23$
for some cluster $i$. Then, the total weight is at least $A_i$.
\end{claim}
\begin{proof} The total size is above $$1+(A_i-2)\cdot \frac
23=2A_i/3-1/3$$ (by considering together the last bin, and the bad
inner bin if it exists or the first inner bin otherwise). Thus,
$W_i \geq 1.8(2A_i/3-1/3)=1.2\cdot A_i-0.6 = A_i+0.2(A_i-3) \geq
A_i$ as $A_i \geq 3$.\end{proof}

Let $\tau=\tau_i$ be the number of bins in the prefix for inner
bins of FFD for cluster $i$ with large items (where $\tau_i=0$ if
this prefix is empty). We have $\tau_i \leq A_i-1$ (because we
only consider inner bins).

\begin{claim}\label{con2}Assume that $\tau_i \geq 2$, and every inner bin whose
index is above $\tau_i$ has load of at least $\frac 23$. Then, the
total weight is at least $A_i$.\end{claim}
\begin{proof} The total size of items is at least $$(\tau-1)\cdot \frac 12
+1 + \frac
23(A_i-1-\tau)=(\tau+1)/2+2(A_i-1-\tau)/3=2A_i/3-\tau/6-1/6$$ (by
considering the first and last bin together, and since $\tau \leq
A_i-1$) the number of inner bins with loads at least $\frac 23$ is
non-negative).

Thus $S_i \geq 2A_i/3-\tau/6-1/6 $, and we have $$W_i \geq
1.8(2A_i/3-\tau/6-1/6)+0.15\tau=1.2 A_i -0.15\tau-0.3\ . $$ As
$\tau \geq 2$ and $A_i \geq \tau+1$ we get $$W_i \geq A_i +0.2
A_i-0.15 \tau -0.3 \geq A_i +
0.2(\tau+1)-0.15\tau-0.3=A_i+0.05\tau-0.1 \geq A_i \ .
$$\end{proof}

Let $\theta$ be the the first item of the last bin of FFD
and its size.
\begin{claim} At least one of the cases of Claims \ref{con1},\ref{con2} holds for every
cluster.\end{claim}
\begin{proof} If $\theta \leq \frac 13$, by the second part of
Claim \ref{was3}, all inner bins have loads above $\frac 23$.

Otherwise, all items that arrived before $\theta$ have sizes above
$\frac 13$. Every inner bin without a large item has exactly two
items of sizes in $(\frac 13,\frac 12]$. If $\tau \leq 1$, the
condition of Claim \ref{con1} holds and otherwise the condition of
Claim \ref{con2} holds.
\end{proof}

\begin{proposition}
The price of clustering is at most 1.95.
\end{proposition}
\begin{proof}
We have $\sum_{i=1}^\ell OPT_i \leq \sum_{i=1}^{\ell} A_i \leq
\sum_{i=1}^{\ell} W_i = W \leq 1.95 \cdot OPT $.
\end{proof}

\medskip

We proceed to an improved analysis. Intuitively, a bad structure
of clusters is that used in the lower bound construction, that is,
clusters consist of items of similar sizes, some of which are slightly smaller than a
given reciprocal of an integer and some slightly larger than this value.
Our improved weight function is based on dealing with such
clusters, and in particular, such clusters for items that are
relatively large.

For simplicity, we will use the same notation, and use $w$ as the
name of our new weight function, and the function is also based on
item sizes $w(x):[0,1)\rightarrow (0,\frac{7297}{3900}\approx
1.871026]$ as follows.
\begin{equation*}
\displaystyle w(x)=\frac {21}{13} \cdot x+\begin{cases}
\frac{997}{3900}\approx 0.255641 {  \ \ \  \ \ \ \ \ \ \  \ \ \ \
\ \ \mbox {for} \ \ \ \ \  \  \ \ \ } x>\frac 12, \vspace{5pt}
\\

\vspace{5pt}

\frac{64}{975}=\frac{256}{3900} \approx 0.065641{ \ \  \ \ \ \ \ \
\mbox {for} \ \ \ \ \  \  \ \ \ } \frac 13 < x
\leq \frac 12, \\
\vspace{5pt}

\frac{18}{325}=\frac{216}{3900} \approx 0.0553846{ \ \ \ \ \  \ \
\mbox {for} \ \ \ \ \  \  \ \ \ } \frac 14 < x
\leq \frac 13, \\
\vspace{5pt}

\frac{2}{195}=\frac{40}{3900}\approx 0.01025641{ \ \ \  \ \ \
\mbox {for}  \ \ \ \ \ \ \  \ \ } \frac 16 < x
\leq \frac 14, \\
\vspace{5pt}

0 { \ \ \ \ \ \ \ \ \ \ \ \ \ \ \ \ \ \ \ \ \  \ \ \ \ \   \ \ \ \
\ \ \ \ \ \ \mbox {for} \ \ \ \ \ \ \ \ \ } 0 \leq x \leq \frac
16.
\end{cases}
\end{equation*}

The values after the equality sign in the definition are called
bonuses, and their values are between zero and $\frac{997}{3900}$.
Next, we find an upper bound on the total weight.
\begin{claim} We have $W \leq \frac{581}{300} \cdot OPT$, where $\frac{581}{300}\approx 1.93667$.\end{claim}
\begin{proof} Consider a bin $\cal{B}$ of OPT. The total size of items
is at most $1$. We consider items of $\cal{B}$ of sizes above
$\frac 1 6$, as the weight of any other item is just $\frac
{21}{13}$ times its size. If there is no large item, the ratio
$\frac{w(x)}{x}$ does not exceed $\frac{597}{325}<1.84$, so the
total weight for $\cal{B}$ is below $1.9$.

If there is a large item, there can be at most two other items of
sizes above $\frac 16$. If there is also an item of size in
$(\frac 13,\frac 12]$, these are the only two items with positive bonuses. Otherwise, if there
is also an item of size in $(\frac 14,\frac 13]$, there can be
another item of size in $(\frac 16,\frac 14]$. In other cases the total weight is smaller.
Thus, the total
weight is at most
$$\frac{21}{13}+\frac{997}{3900}+\max\{\frac{256}{3900},\frac{216}{3900}+\frac{40}{3900}\}=\frac{21}{13}+\frac{1253}{3900}=\frac{7553}{3900}=\frac{581}{300} \ . $$
\end{proof}

Once again, we show that for every cluster, it holds that $W_i
\geq A_i$. We use the index $\tau$ and the size $\theta$ as
before.

\begin{lemma}
Given a cluster $i$, it holds that $W_i \geq A_i$.
\end{lemma}
\begin{proof}
If the load of any inner bin, possibly excluding one inner bin, is
at least $\frac 67$, since the total load of any inner bin and the
last inner bin is above $1$, we get $$S_i \geq 1 + \frac
67(A_i-2)=\frac{6A_i-5}7 \ . $$ Thus, $$W_i\geq \frac {21}{13}
\cdot \frac{6A_i-5}7=\frac{18A_i-15}{13}=A_i+5\cdot
\frac{A_i-3}{13} \geq A_i \ , $$ since $A_i \geq 3$. Thus, we
assume that at least two inner bins have loads below $\frac 67$.

We split the analysis into  several cases. In the case $\theta \leq \frac 17$, the
load of any inner bin is above $1-\theta \geq \frac 67$, so this
case was already excluded.

In the case $\theta \in (\frac 17,\frac 16]$, the load of every
inner bin is at least $\frac 56$, and the total size satisfies
$$S_i \geq 1+\frac 56 (A_i-2)=\frac {5A_i-4}6 \ . $$ We have $$\frac
{21}{13}\cdot S_i \geq 7 \cdot
\frac{5A_i-4}{26}=A_i+\frac{9A_i-28}{26}\geq A_i-\frac
{150}{3900} \ , $$ by $A_i \geq 3$. If there is at least one item of
size above $\frac 14$, or at least four items with positive
bonuses, we are done. Otherwise, there are at most three items of sizes
above $\frac 16$, all of which are not larger than $\frac 14$, and every inner bin, except for possibly the
first one, has six items of sizes in $(\frac 17,\frac 16]$, so the
loads are above $\frac 67$, contradicting the assumption.

In the case $\theta \in (\frac 16,\frac 15]$, the load of every
inner bin is at least $\frac 45$, and the total size satisfies
$$S_i \geq 1+\frac 45(A_i-2)=\frac {4A_i-3}5 \ . $$ We have $$\frac
{21}{13}\cdot S_i \geq 21 \cdot
\frac{4A_i-3}{65}=A_i+\frac{19A_i-63}{65}\geq A_i-\frac
{360}{3900} \ , $$ by $A_i \geq 3$. The last bin has an item of bonus
$\frac {40}{3900}$. Every inner bin with an item of size above
$\frac{1}{4}$ has a bonus of at least $\frac{216}{3900}$, and
every inner bin without such an item has at least four items of
sizes in $(\frac 16,\frac 14]$, so the total bonus is at least
$\frac{160}{3900}$. The total bonus is therefore at least
$\frac{360}{3900}$, and $W_i \geq A_i$.

In the case $\theta \in (\frac 15,\frac 14]$, the load of every
inner bin is at least $\frac 34$, and the total size satisfies
$$S_i \geq 1+\frac 34(A_i-2)=\frac {3A_i-2}4 \ . $$ We have $$\frac
{21}{13}\cdot S_i \geq 21 \cdot
\frac{3A_i-2}{52}=A_i+\frac{11A_i-42}{52}\geq A_i-\frac
{675}{3900} \ , $$ by $A_i \geq 3$. If there is a large item, we are
done. Otherwise, if there are at least three items of sizes in
$(\frac 14,\frac 12]$, since the last bin has an item of bonus
$\frac {40}{3900}$, the total bonus is at least
$\frac{688}{3900}$. Otherwise, all items of sizes above $\frac 14$
(at most two such items) are packed into the first inner bin, and every inner bin except
for the first one has at least four items of sizes in $(\frac 15,
\frac 14]$, and its load is above $\frac 45$. As in the case
$\theta \in (\frac 16,\frac 15]$, we have $\frac {21}{13}\cdot
S_i\geq A_i-\frac {360}{3900}$, and the calculation of bonuses it
also the same as in that case.

In the case $\theta \in (\frac 14,\frac 13]$, the load of every
inner bin is at least $\frac 23$, and the total size satisfies
$S_i \geq 1+\frac 23  (A_i-2)=\frac {2A_i-1}3$. We have
$$\frac{21}{13} \cdot
S_i=\frac{14A_i-7}{13}=A_i+\frac{A_i-7}{13}\geq
A_i-\frac{1200}{3900} \ .$$ The bonus of the item of the last bin is
$\frac{216}{3900}$, so if there is also a large item, the total
bonus is at least
$\frac{216}{3900}+\frac{997}{3900}=\frac{1213}{3900}$ and we are
done. Otherwise, if all inner bins together have at least five items of sizes
in $(\frac 14,\frac 12]$, the total bonus is at least
$\frac{1296}{3900}$. We are left with the case where there are at
most four such items, and in fact there are exactly four such items,
as every inner bin has at least two such items. If the second
inner bin does not have an item of size above $\frac 13$, then it
has three items of sizes above $\frac 14$, so this case is
impossible. Thus, the first bin has two items of sizes above
$\frac 13$, and the second bin has at least one such item. The
total bonus in this case is $3\cdot \frac{256}{3900}+2\cdot
\frac{216}{3900}=\frac{1200}{3900}$.

In the case $\theta > \frac 13$, we use the value of $\tau$ in the
analysis. We consider the first inner bin together with the last
bin. Every inner bin that is not in the prefix of first $\tau$
inner bins has two items of sizes in $(\frac 13,\frac 12]$. The
bins of indices $2,\ldots,\tau$ have large items, and the first
inner bin either has a large item or two items of sizes in $(\frac
13,\frac 12]$, so its bonus is at least $\frac{512}{3900}$.

If $\tau \leq 1$, we get $S_i \geq 1+(A_i-2)\frac
23=\frac{2A_i-1}3$ and $$W_i \geq \frac {21}{13} \cdot
\frac{2A_i-1}3 +(2A_i-1)\frac {256}{3900}=\frac {4712}{3900}
A_i-\frac {2356}{3900} =A_i+\frac{812A_i-2356}{3900}\geq
A_i+\frac{2436-2356}{3900} > A_i  \ . $$

Otherwise, $\tau \geq 2$, and the first inner bin has a large
item. Thus, $$S_i \geq 1+(\tau-1)\cdot \frac 12+(A_i-\tau-1)\cdot
\frac 23=\frac{2A_i}3-\frac 16-\frac{\tau}6 \ , $$ and
$$W_i \geq \frac
{21}{13}\cdot(\frac{2A_i}3- \frac
16-\frac{\tau}6)+(2(A_i-\tau-1)+1)\cdot
\frac{256}{3900}+\tau\cdot\frac{997}{3900}=\frac{4712}{3900}\cdot
A_i-\frac{565\tau}{3900}-\frac{1306}{3900}\ . $$
By $\tau \leq A_i-1$,
we have $$W_i \geq \frac{4712}{3900}\cdot
A_i-\frac{565(A_i-1)}{3900}-\frac{1306}{3900}
=A_i+\frac{247A_i-741}{3900}\geq A_i \ . $$
\end{proof}

\begin{theorem}
The price of clustering is at most $\frac{581}{300}\approx
1.93667$.
\end{theorem}
\begin{proof}
We have $\sum_{i=1}^\ell OPT_i \leq \sum_{i=1}^{\ell} A_i \leq
\sum_{i=1}^{\ell} W_i = W \leq \frac{581}{300} \cdot OPT $.
\end{proof}

\subsection{The price of clustering for larger parameters $\boldsymbol{k\geq 4}$}
\label{xxx} In this section we briefly discuss the case of larger
$k$, that is, the case where for a given integer $k\geq 4$, it is
known that the optimal solution for every cluster has cost not
smaller than $k$.

The lower bound has a similar structure in the sense that items
type are similar. In the case $k=3$, half of the bins of a
globally optimal solution two items of sizes close to $\frac 12$,
while here only a fraction of $\frac{1}{k-1}$ of the bins will be
such. In the clustered solution, clusters with items of sizes
approximately $\frac 13$ will still have two items of sizes below
$\frac 13$, but the number of items of sizes above $\frac 13$ will
be $2k-3$. For clusters with items of sizes approximately $\frac
16$, there are still five items of size below $\frac 16$ in each
such cluster, but there are $5k-9$ items of sizes above $\frac
16$. For items of sizes just above $\frac 1{t}$ ($t=7,43,1807$),
still the last bin of the cluster will have just one item, but
there are $k$ bins, so the number of items will be $k\cdot
(t-1)-(t-2)=(k-1)\cdot (t-1)+1$.

The resulting numbers of items (up to negligible constants) are as
follows. The number of items of sizes just above $\frac 12$ is $N$
and the number of items of sizes just below $\frac 12$ is
$\frac{N}{k-1}$, items of sizes just above $\frac 13$:
$\frac{(2k-3)(k-2)N}{(k-1)(2k-1)}$, items of sizes just below
$\frac 13$: $\frac{2(k-2)N}{(k-1)(2k-1)}$, items of sizes just
above $\frac 16$: $\frac{2(k-2)N}{(k-1)(2k-1)}$, items of sizes
just below $\frac 16$: $\frac{10(k-2)N}{(k-1)(2k-1)(5k-9)}$, and
items of sizes just above $\frac 17$, $\frac 1{43}$, and
$\frac{1}{1807}$:
$\frac{(k-2)(10k^2-33k+17)N}{(k-1)(2k-1)(5k-9)}$.

\begin{proposition}
The lower bound on the price of clustering for a given value
$k\geq 4$ is $$\frac
{k}{k-1}+\frac{k(k-2)}{(k-1)(2k-1)}+\frac{2k(k-2)}{(k-1)(2k-1)(5k-9)}$$
$$+\frac{k(10k^3-53k^2+83k-34)(\frac 1{6k-5}+\frac 1{42k-41}+\frac 1{1806k-1805})}{(k-1)(2k-1)(5k-9)}  \ . $$
\end{proposition}

Since this generalizes the case $k=3$, indeed for $k=3$ we get the
earlier lower bound of $1.9355858244424$. For $k=4,5,6,7,8,9$, and
$10$, the approximate lower bounds are 1.8781318, 1.8410851,
1.815945, 1.7979, 1.78437, 1.77386, and 1.76546, respectively.

The lower bound for $k$ growing to infinity is only approximately
$1.6910299$ since we did not use the entire series but only the
first few elements of the sequence $c_i$ defined earlier.

\medskip

It is possible to show close upper bounds for other values of $k$
as well. As an example, we show a close upper bound for $k=4$.
Once again, we expect the worst case for clusters to be of the
same form as before, but there will be another relatively full bin
in every cluster.

We will use the same notation once more, and use $w$ as the name
of our new weight function, and the function is also based on item
sizes. Let $\Delta=77805$.
\begin{equation*}
\displaystyle w(x)=\frac {28}{19} \cdot x+\begin{cases} \alpha=
\frac{25124}{\Delta}\approx 0.32291 {  \ \ \  \ \ \ \ \ \ \ \ \ \
\ \mbox {for} \ \ \ \  \  \ \ \ } x>\frac 12, \vspace{5pt}
\\

\vspace{5pt}

\beta=\frac{6528}{\Delta}\approx 0.083902
 { \  \ \ \  \ \ \  \ \ \  \ \ \  \ \mbox {for} \ \
\ \ \ \ \ \ } \frac 13 < x
\leq \frac 12, \\
\vspace{5pt}

\gamma=\frac{5520}{\Delta}\approx 0.0709466{ \ \ \ \ \  \  \  \ \
\ \ \ \  \mbox {for} \ \ \ \ \ \ \  \ } \frac 14 < x
\leq \frac 13, \\
\vspace{5pt}

\delta=\frac{1008}{\Delta}\approx 0.0129555{   \ \ \  \ \  \ \ \ \
\ \ \ \ \mbox {for}  \ \ \ \ \ \  \  \ } \frac 16 < x
\leq \frac 14, \\
\vspace{5pt}

0 { \ \ \  \ \ \ \ \   \ \ \  \ \ \ \ \ \   \ \ \ \ \ \ \ \ \ \ \
\ \ \ \ \ \ \  \ \ \  \  \mbox {for}  \ \  \ \ \ \ \ \ } 0 \leq x
\leq \frac 16.
\end{cases}
\end{equation*}

Let $\lambda=146312$, where $\frac{\lambda}{\Delta}\approx
1.880496112076$. We have $\frac{28}{19}\approx 1.47368421$.

The values after the equality sign in the definition are called
bonuses again. We find an upper bound on the total weight.
\begin{claim} We have $W \leq \frac{\lambda}{\Delta} \cdot OPT$.\end{claim}
\begin{proof} Consider a bin $\cal{B}$ of OPT. The total size of items
is at most $1$. We consider items of $\cal{B}$ of sizes above
$\frac 1 6$, as the weight of any other item is just $\frac
{28}{19}$ times its size. If there is no large item, the ratio
$\frac{w(x)}{x}$ does not exceed $1.8$, so the total weight for
$\beta$ is below $\frac{\lambda}{\Delta}$. The remaining case,
similarly to the calculation for $k=3$ yields
$\frac{28}{19}+\alpha+\max\{\beta,\gamma+\delta\}=\frac{\lambda}{\delta}$.
\end{proof}

Once again, we show that for every cluster, it holds that $W_i
\geq A_i$. We use the index $\tau$ and the size $\theta$ as
before.

\begin{lemma}
Given a cluster $i$, it holds that $W_i \geq A_i$.
\end{lemma}
\begin{proof}
If the load of any inner bin, possibly excluding one inner bin, is
at least $\frac 67$, once again we get $S_i \geq (A_i-2)\cdot
\frac 67+ 1=\frac{6A_i-5}7$. Thus, $$W_i\geq \frac {28}{19} \cdot
\frac{6A_i-5}7=\frac{24A_i-20}{19}=A_i+5\cdot \frac{A_i-4}{19}
\geq A_i \ ,$$ since $A_i \geq 4$. Thus, we assume that at least
two inner bins have loads below $\frac 67$, and therefore we can
 assume that $\theta > \frac 17$ holds again.

In the case $\theta \in (\frac 17,\frac 16]$, we found $S_i \geq
\frac {5A_i-4}6$, and we have $$\frac {28}{19}\cdot S_i \geq 14
\cdot \frac{5A_i-4}{57}=A_i+\frac{13A_i-56}{57}\geq A_i-\frac
{4}{57} $$ by $A_i \geq 4$. If $A_i \geq 5$, we already get $\frac
{28}{19}\cdot S_i \geq A_i$, so we focus on the case $A_i=4$ for
the current range of $\theta$.

As $\gamma>\frac 4{57}$ and $6 \cdot \delta >\frac 4{57}$, if
there is at least one item of size above $\frac 14$, or at least
six items with positive bonuses, we are done. Otherwise, there are
at most five items of sizes above $\frac 16$, all of which are not
larger than $\frac 14$, so the third inner bin has six items of
sizes in $(\frac 17,\frac 16]$, and more specifically, all six
items of sizes at least $\theta$. This gives us another lower
bound on the total size: $$S_i \geq
(A_i-2)(1-\theta)+6\theta+\theta=2+5\theta \geq
2+\frac{5}{7}=\frac{19}{7} \ , $$ and $\frac{28}{19} S_i \geq 4$.

In the case $\theta \in (\frac 16,\frac 15]$, the total size
satisfies $S_i \geq \frac {4A_i-3}5$, and we have
$$\frac {28}{19}\cdot S_i \geq 28 \cdot
\frac{4A_i-3}{95}=A_i+\frac{17A_i-84}{95}\geq A_i-\frac {16}{95}$$
by $A_i \geq 4$. The last bin has an item of bonus $\delta$. Every
inner bin with an item of size above $\frac{1}{4}$ has a bonus of
at least $\gamma$, and every inner bin without such an item has at
least four items of sizes in $(\frac 16,\frac 14]$, so the total
bonus is at least $4\delta$. The total bonus is therefore at least
$13\delta=\frac{16}{95}$, and $W_i \geq A_i$.

In the case $\theta \in (\frac 15,\frac 14]$, the total size
satisfies $S_i \geq \frac {3A_i-2}4$. We have $$\frac
{28}{19}\cdot S_i \geq 7 \cdot
\frac{3A_i-2}{19}=A_i+\frac{2A_i-14}{19}\geq A_i-\frac {6}{19} \ ,
$$ by $A_i \geq 4 $.

If there is a large item, we are done as $\alpha>\frac 6{19}$.
Otherwise, since $5\cdot c \geq \frac{6}{19}$, if there are at
least five items of sizes in $(\frac 14,\frac 12]$, we are done.
Otherwise, all items of sizes above $\frac 14$ (at most four such
items) are packed into the two first inner bins, and every inner
bin except for the first two has at least four items of sizes in
$(\frac 15, \frac 14]$, and its load is above $\frac 45$. If there
is just one inner bin with at least one item of size above $\frac
14$, we get the same bound on the total size and the entire
analysis is the same as in the case where $\theta \in (\frac
16,\frac 15]$. Otherwise, there are at least three items with
bonuses of $\gamma$ (two in the first bin and one in the second
bin), and the second bin has at least one additional item with a
positive bonus. There are at least five other items (including an
item of the last bin) with bonuses of $\delta$. The total bonus is
therefore at least $3\gamma+6\delta>0.29$. The total size is at
least $(A_i-3)\cdot 0.8
+0.75+1=\frac{16A_i-48+15+20}{20}=\frac{16A_i-13}{20}$, and
$$\frac{28}{19} \cdot S_i \geq 7 \cdot \frac{16A_i-13}{95} =
A_i+\frac{17A_i-91}{95}\geq A_i-\frac{23}{95} > A_i -0.25 \ .$$

In the case $\theta \in (\frac 14,\frac 13]$, the total size
satisfies $S_i \geq \frac {2A_i-1}3 $, and we have $\frac{28}{19}
\cdot S_i=\frac{56A_i-28}{57}=A_i-\frac{A_i+28}{57}$. Every inner
bin has one, or two, or three items of sizes above $\frac 14$. In
the case of one item, it is large and its bonus is $a$. In the
case of two items, at least one of them has size above $\frac 13$,
and if there are multiple such bins, only one of the bins with two
items has an item with size at most $\frac 13$ as its second item.
In the case of three items, the total bonus is at least $3\gamma$.
Thus, $A_i-2$ bins have bonuses of at least $2\beta$, the last bin
has a bonus of at least $\gamma$, and another bin has a bonus of
at least $\beta+\gamma$. Thus, the total bonus is at least
$(2A_i-3)\beta+2\gamma\geq \frac{A_i+28}{57}$, since this is
equivalent to $(2\beta-\frac1{57})A_i \geq
\frac{28}{57}+3\beta-2\gamma$, which holds for $A_i\geq 4$.

In the case $\theta > \frac 13$, we use the value of $\tau$ in the
analysis. We consider the first inner bin together with the last
bin. Every inner bin that is not in the prefix of first $\tau$
inner bins has two items of sizes in $(\frac 13,\frac 12]$. The
bins of indices $2,\ldots,\tau$ have large items, and the first
inner bin either has a large item or two items of sizes in $(\frac
13,\frac 12]$, so its bonus is at least $2\beta$.

If $\tau \leq 1$, we get $S_i \geq 1+(A_i-2)\frac
23=\frac{2A_i-1}3$ and $$W_i \geq \frac {28}{19} \cdot
\frac{2A_i-1}3 +(2A_i-1)\beta=(\frac{56}{57}+2\beta)
A_i-\frac{28}{57}-\beta =A_i+(2\beta-\frac{1}{57})A_i
-(\frac{28}{57}+\beta)> A_i  \ . $$ Since $2\beta>\frac 1{57}$ and
$A_i \geq 4$, this is at least $$A_i +4(2\beta-\frac
1{57})-(\frac{28}{57}+\beta)=A_i+7\beta-\frac{32}{57}>A_i \ . $$

Otherwise, $\tau \geq 2$, and the first inner bin has a large
item, and $S_i \geq \frac{2A_i}3-\frac 16-\frac{\tau}6$,  and
$$W_i \geq \frac {28}{19}\cdot(\frac{2A_i}3- \frac
16-\frac{\tau}6)+(2(A_i-\tau-1)+1)\cdot
\beta+\tau\cdot\alpha=(\frac{56}{57}+2\beta)\cdot
A_i-(\frac{14}{57}+2\beta-\alpha )\tau-(\frac{14}{57}+\beta) \ .
$$ Since $\alpha<\frac{14}{57}+2\beta$, and by $\tau \leq A_i-1$,
we have $$W_i \geq (\frac{56}{57}+2\beta)\cdot
A_i-(\frac{14}{57}+2\beta-\alpha
)(A_i-1)-(\frac{14}{57}+\beta)=(\frac{42}{57}+\alpha)A_i+(\beta-\alpha)$$
$$=A_i+(\alpha-\frac{15}{57})A_i+(\beta-\alpha)\geq
A_i+4\alpha-\frac{60}{57}+\beta-\alpha=A_i \ .$$
\end{proof}

\medskip

We conclude with the following.

\begin{theorem} The price of clustering is at most
$\frac{\lambda}{\Delta} \approx 1.88049612$. \end{theorem}

\section{Bin packing with delays}

We briefly discuss assumptions on delay functions. A delay
function $d:[0,\infty)\rightarrow[0,\infty)$ is assumed to be
continuous. We also assume $d(0)=0$ without loss of generality, as
otherwise any algorithm will pay a delay of $d(0)$ and the delay
cost can be modified by subtracting the value $d(0)$ from it. It
is also assumed that the function is monotonically non-decreasing
and unbounded (see below for a short discussion of the unbounded
case). Since a general function can be given by an oracle while
algorithms assume that the value for every time is known
precisely, there will be a small loss in the competitive ratio,
where the loss is small due to continuity.

Our algorithm is a variant of the algorithm of \cite{AESV}, where
this class of algorithms waits until the current total delay
reaches a certain value. For linear delay functions, one can
implement it exactly, while for other delay functions it is
necessary to query the oracle frequently to see whether the
required total delay was already reached. As mentioned in
\cite{AESV}, it is not hard to adapt the previously known
algorithm (presented there) to work with such a delay function,
and with item specific delay functions. For simplicity, we will
describe the algorithm assuming that it is possible to calculate
the current delay for every item exactly, and to maintain this
value in a continuous manner.

Another cause of a small error is due to the usage of an
irrational parameter, which is rounded slightly. The parameter of
\cite{AESV} was $1$, so this minor difficulty did not exist there.
For simplicity, we will assume in what follows that any real
parameter can be used exactly.

The algorithm assumes that a total order is given on the arriving
items. The order satisfies the property that an earlier item has
an index smaller than that or an item arriving later. For items
arriving at the same time, an order in which the algorithm
processes them is used. We note that one can assume that an
optimal solution also processes the input as a sequence, and it
for every bin it opens, this is done right after it processes the
last item of the bin (where the last item is the item of maximum
index packed into the bin according to the ordering of the
algorithm). This associates every bin of the optimal solution (a
fixed optimal offline solution which we consider and compare
online algorithms to) with one specific item.

The algorithm has a positive parameter $\rho$ and acts as follows.
The algorithm works in phases, where in every phase it
continuously keeps a value that is the total current delay of all
unpacked items. Once this value reaches $\rho$, the algorithm
defines the current phase as the consecutive subsequence of items
starting with the first item that does not belong to the previous
phase (or starting with the very first item, if this is the first
phase), and ending with the last item that was already processed.
It packs the items of the current phase by FFD, and it will start
a new phase with the next item, if it exists, or it will terminate
if the input ended. We call this algorithm {\it modified} since
the main structure is unchanged and it is the same as the one of
\cite{AESV}, but we use a parameter $\rho>0$, while the parameter
of \cite{AESV} was simply equal to $1$. Our analysis will be
different.

Note that even if no new items arrive and the input was
terminated, the algorithm may still be in the process of
constructing the last phase. Since the delay functions are
unbounded, the algorithm will pack all items of the last phase
once the last phase is defined, and this will happen before it
halts. Alternatively, it is possible to use bounded delay
functions. In this case, it could happen in the last phase that
the total delay will not reach the value $\rho$, and the algorithm
should pack the remaining items once the input has stopped.

Our parameter $\rho$ will be equal to approximately
$0.4640251938$.

\begin{theorem} The competitive ratio of the modified algorithm with the best parameter is at most $3.1550554008$.
\end{theorem}
\begin{proof}  Let $I$ be the input, let $\ell$ be the number of phases used by the
algorithm, let $X_i \geq 1$ be the number of bins of phase $i$
(where every phase has at least one item), and let $I_i$ be the
set of items for this phase. We use the following analysis of FFD.
For a set of items $J$, let $FFD(J)$ be the number of bins that
FFD creates for $J$, and let $V(J)$ be the total weight of items
for a given weight function $v$.

By \cite{EL06a}, Lemmas 2 and 5, there exists a weight function $v:[0,1]\rightarrow
[0,1]$ such that the two properties below  hold for any set of
items $J$.

\smallskip

\noindent {\bf 1.} $FFD(J) \leq V(J)+1$.

\smallskip

\noindent {\bf 2.} For any set of items $J'$ of total size at most $1$, $V(J')
\leq \pi_{\infty}$, where $\pi_{\infty} < 1.691030207$.

\smallskip

The definition of $v$ is as follows. For $x>\frac 12$, it holds
that $v(x)=1$. For $0<x\leq \frac 12$, such that $(\frac
1{j+1},\frac 1j]$ for an integer $j$ (there is a unique value of
$j$ for every such $x$), $v(x)=x+1/(j(j+1))$. Finally, for $x=0$,
let $v(x)=0$. This last value does not appear in \cite{EL06a}, and
we briefly explain why it does not affect the properties. For the
first property, if $J$ only has items of size zero, we have
$V(J)=0$ and $FFD(J)=1$. Otherwise, the output of FFD and the
total weight is not affected by items of size zero. As for the
second property, adding items of size zero changes neither the
total size nor the weight.

The value $\pi_{\infty}$ is defined by a sequence frequently
encountered in bin packing problems \cite{LeeLee85}, defined as
follows: $c_1=1$, and for $i>1$, $c_i=c_{i-1}(c_{i-1}+1)$. Then,
$\pi_{\infty}=\sum_{i=1}^{\infty} \frac{1}{c_i}$.

Using the first property for every phase separately, we get that
$X_i-1 =FFD(I_i) - 1 \leq V(I_i)$ holds for every phase $i$. The
cost of the algorithm for phase $i$ is $\rho+X_i \leq
\rho+1+V(I_i) $ For all phases, we have the the cost of the
algorithm is $$\sum_{i=1}^{\ell}
(\rho+1+V(I_i))=\ell\cdot(\rho+1)+\sum_{i=1}^{\ell}
V(I_i)=\ell\cdot (\rho+1)+V(I) \ . $$

Consider a fixed optimal solution for the input, denoted by $OPT$.
Let $B$ and $D$ denote the number of bins of $OPT$, and the delay
of $OPT$, respectively. By the second property of the weight
function, $V(I) \leq \pi_{\infty}\cdot B$.

There can be two types of phases. The first type is a phase for
which the optimal solution has at least one bin that is associated
with an item of the phase. The second type is a phase where there
is no bin of the optimal solution that is associated with an item
of the phase. Let $\ell_1$ and $\ell_2$ (where
$\ell=\ell_1+\ell_2$) be the numbers of phases of the two types.
Obviously, it holds that $B \geq \ell_1$. In phases of the second
type, $OPT$ pays at least the same delay as the algorithm (since
the same items wait at least the same time to be packed), so $D
\geq \ell_2\cdot \rho$.

We get $$(1+\rho)\ell = (1+\rho)\ell_1+(1+\rho)\ell_2 \leq
(1+\rho)B+(1+\frac 1{\rho})D \ . $$

Thus, the total cost of the algorithm is at most
$$\sum_{i=1}^{\ell} (\rho+X_i)=\ell\cdot(\rho+1)+\sum_{i=1}^{\ell}
(X_i-1)\leq (1+\frac 1{\rho})D+(1+\rho)B+B\cdot
\pi_{\infty}=(1+\frac 1{\rho})D+(1+\rho+\pi_{\infty})B \ . $$

Letting $\rho \approx 0.4640251938$ we get a competitive ratio not
exceeding the following bound: \\ ${\max\{(1+\frac
1a),1+a+\pi_{\infty}\}<3.1550554008}$.
\end{proof}

\bibliographystyle{abbrv}

\end{document}